\documentclass[submission,copyright,creativecommons]{eptcs}
\providecommand{\event}{AFL 2017} 
\usepackage{underscore}           

\usepackage{amsmath,mathtools,amssymb,amsthm}
\newtheorem{theorem}{Theorem}
\newtheorem{lemma}[theorem]{Lemma}
\newtheorem{corollary}[theorem]{Corollary}
\theoremstyle{definition}

\theoremstyle{remark}

\title{The Triple-Pair Construction for Weighted $\omega$-Pushdown Automata}
\author{Manfred Droste
\institute{Universit\"at Leipzig, Institut f\"ur Informatik,\\ Germany}
\email{droste@informatik.uni-leipzig.de}
\and
Zolt\'an \'Esik\footnote{Zolt\'an \'Esik died on May 25, 2016.}
\institute{University of Szeged,\\ Department of Foundations of Computer Science,\\ Hungary}
\and
Werner Kuich
\institute{Technische Universit\"at Wien, \\
        Institut f\"ur Diskrete Mathematik und Geometrie,\\ Austria}
\email{kuich@tuwien.ac.at}
}
\def\titlerunning{The Triple-Pair Construction for Weighted $\omega$-Pushdown Automata}
\def\authorrunning{M. Droste, Z. \'Esik \& W. Kuich}

\newcommand*{\Sr}{S}

\newcommand*{\Pmc}{\mathcal{P}}

\newcommand*{\Sn}{V^n}
\newcommand*{\Snn}{\Sr^{n \times n}}

\newcommand*{\Spnn}{{\Sr'}^{n \times n}}
\newcommand*{\Snngam}{(\Snn)^{\Gamma^* \times \Gamma^*}}
\newcommand*{\Spnngam}{(\Spnn)^{\Gamma^* \times \Gamma^*}}
\newcommand*{\llss}{\ll \Sigma^* \gg}
\newcommand*{\llso}{\ll \Sigma^\omega \gg}
\newcommand*{\ssigep}{S \langle \Sigma \cup \{\epsilon\} \rangle}
\newcommand\mydots{\ifmmode\ldots\else\makebox[0.8em][c]{.\hfil.\hfil.}\fi}
\newcommand*{\bhvr}[1]{\Vert #1 \Vert}

\renewcommand{\epsilon}{\varepsilon}

\newcommand{\N}{\mathbb{N}}
\newcommand*{\Ninf}{\N^\infty}
\newcommand{\B}{\mathbb{B}}
\newcommand{\Pcc}{\mathcal{P}}
\newcommand*{\ipj}{[i,p,j]}

\newcommand*{\Moml}{M^{\omega,l}}

\newcommand*{\mpom}{[m_1,p_0,m_2]}
\newcommand*{\mpim}{[m_1,p_1,m_2]}

\newcommand*{\ip}{[i,p]}

\newcommand*{\derl}{\Rightarrow_{\!L}}
\newcommand*{\derls}{\Rightarrow_{\!L}^*}

\newcommand*{\Nlse}{\N^\infty \langle \Sigma \cup \{\epsilon\} \rangle}

\makeatletter
\newcommand*{\rom}[1]{\expandafter\@slowromancap\romannumeral #1@}
\makeatother

\begin{document}
\maketitle

\begin{abstract}
Let $\Sr$ be a complete star-omega semiring and $\Sigma$ be an alphabet.
For a weighted $\omega$-pushdown automaton  $\Pcc$
with stateset $\{1, \dots, n\}$, $n \geq 1$, we show that
there exists a mixed algebraic system over a complete semiring-semimodule pair
${((\Sr \ll \Sigma^* \gg)^{n\times n},  (\Sr \ll \Sigma^{\omega}\gg)^n)}$ such
that the behavior $\bhvr{\Pcc}$ of $\Pcc$ is a component of a solution of this system.
In case the basic semiring is $\B$ or $\N^{\infty}$  we show that there exists
a mixed context-free grammar that generates $\bhvr{\Pcc}$.
The construction of the mixed context-free grammar
from $\Pcc$ is a generalization of the well known triple
construction and is called now triple-pair construction for $\omega$-pushdown automata.
\end{abstract}

\section{Introduction and preliminaries}

Weighted pushdown automata were introduced by Kuich, Salomaa
\cite{88}.
Many results on classical pushdown automata and
context-free grammars can be generalized to
weighted pushdown automata and algebraic systems.
Classic pushdown automata can also be used to accept infinite
words (see Cohen, Gold \cite{23})
and it is this aspect we generalize in our paper.
We consider weighted $\omega$-pushdown automata and their
relation to algebraic systems over a complete
semiring-semimodule pair $(\Snn, \Sn)$.
It turns out that the well known triple construction for
pushdown automata can be generalized to a triple-pair construction for weighted $\omega$-pushdown automata.
Our paper generalizes results of Droste, Kuich \cite{roc}.

The paper consists of this and three more sections.
In Section \ref{sec:2}, pushdown transition matrices are introduced and their properties are studied.
The main result of this section is that, for such a matrix $M$, the $p$-blocks, $p$ a pushdown symbol, of the infinite column vector $\Moml$ satisfy a special equality.
In Section \ref{sec:3}, weighted $\omega$-pushdown automata are introduced.
We show that for a weighted $\omega$-pushdown automaton
$\Pcc$ there exists a mixed algebraic system such that
the behavior $\bhvr{\Pcc}$ of $\Pcc$ is a component of a
solution of this system.
In Section \ref{sec:4} we consider the case that the complete
star-omega semiring $\Sr$ is equal to $\B$ or $\Ninf$.
Then for a given weighted $\omega$-pushdown automaton $\Pcc$
a mixed context-free grammar is constructed that generates
$\bhvr{\Pcc}$.
The construction is a generalization of the well known triple
construction and is called \emph{triple-pair construction
for $\omega$-pushdown automata}.

For the convenience of the reader, we quote definitions and results of \'Esik, Kuich \cite{42,43,44,45} from \'Esik, Kuich \cite{MAT}. The reader should be familiar with Sections 5.1-5.6 of \'Esik, Kuich \cite{MAT}.

A semiring $S$ is called \emph{complete starsemiring} if sums for all
families $(s_i \mid i \in I)$ of elements of $S$ are defined, where $I$ is an arbitrary index set, and if S is
equipped with an additional unary star operation
$^*: \Sr \to \Sr$ defined by $s^* = \sum_{j\geq 0} s^j$ for all $s \in \Sr$.
Moreover, certain conditions have to be satisfied making sure that computations with ``infinite'' sums can be performed analogous to those with finite sums.

A pair $(S,V)$, where $S$ is a complete starsemiring and $V$ is a complete $\Sr$-semimodule is called a
\emph{complete semiring-semimodule pair} if products for
all sequences $(s_i \mid i \in \N)$ of elements of $S$ are
defined and if $S$ and $V$ are equipped with an omega
operation $^\omega: \Sr \to V$ defined by
$s^\omega = \prod_{j\geq1} s$ for all $s \in \Sr$.
Moreover, certain conditions
(e.g. ``infinite'' distributive laws) have to be satisfied making sure that computations with  ``infinite'' sums and  ``infinite'' products can be performed analogous to those with finite sums and finite products.
(For details see Conway \cite{25}, Eilenberg \cite{29}, Bloom, \'Esik \cite{10}, \'Esik, Kuich \cite{MAT}, pages 30 and 105-107.)

A semiring $\Sr$ is called
\emph{complete star-omega semiring} if
$(S,S)$ is a complete semiring-semimodule pair.

For the theory of infinite words and finite automata accepting infinite words by the B\"uchi condition consult Perrin, Pin \cite{PerPin}.

\section{Pushdown transition matrices} \label{sec:2}

In this section we introduce pushdown transition matrices
and study their properties.
Our first theorem generalizes Theorem 10.5 of
Kuich, Salomaa \cite{88}.
Then we show in Theorems~\ref{thm:3.4}~and~\ref{thm:3.8} that, for a pushdown transition matrix $M$,
$(M^\omega)_p$ and  $(\Moml)_p$, $0 \leq l \leq n$, $p \in \Gamma$, introduced below satisfy the
same specific equality.
In Theorem \ref{thm:1}, $\Sr$ denotes a complete starsemiring; afterwards in this section, $(S,V)$ denotes a complete semiring-semimodule pair.

Following Kuich, Salomaa \cite{88} and Kuich \cite{78}, we introduce
pushdown transitions matrices.
Let $\Gamma$ be an alphabet,
called \emph{pushdown alphabet} and let $n\geq 1$.
A matrix $M \in \Snngam$ is termed
a \emph{pushdown transition matrix} (with \emph{pushdown alphabet} $\Gamma$ and \emph{stateset} $\{1,\dots,n\}$) if
\begin{itemize}
        \item[(i)] for each $p \in \Gamma$ there exist only finitely many blocks $M_{p,\pi}$, $\pi \in \Gamma^*$, that are unequal to $0$;
        \item[(ii)] for all $\pi_1, \pi_2 \in \Gamma^*$,
                \begin{equation*}
                M_{\pi_1,\pi_2} = \left\{
                        \begin{array}{ll}
                        M_{p,\pi} & \hspace{0,2cm} \text{if there exist } p \in \Gamma, \pi,\pi' \in \Gamma^* \text{ with } \pi_1 = p\pi' \text{ and } \pi_2 = \pi \pi', \\
                        0 & \hspace{0,2cm} \text{otherwise.}
                        \end{array}
                        \right.
                \end{equation*}
\end{itemize}

For the remaining of this paper, $M \in \Snngam$
will denote a pushdown transition matrix with pushdown alphabet $\Gamma$ and stateset $\{1, \dots, n\}$.

Our first theorem generalizes Theorem 10.5 of Kuich, Salomaa \cite{88} and Theorem 6.2 of Kuich \cite{78} to complete starsemirings.
First observe that for all $\rho_1 \in \Gamma^+$, $\rho_2, \pi \in \Gamma^*$, we have $M_{\rho_1\pi, \rho_2\pi}= M_{\rho_1,\rho_2}$.

Intuitively, our next theorem states that, emptying the pushdown tape with contents $p\pi$ by finite computations has the same effect (i.e., $(M^*)_{p\pi,\epsilon}$) as emptying first the pushdown tape with contents $p$ (i.e., $(M^*)_{p,\epsilon}$) by finite computations and afterwards (i.e., multiplying) emptying the pushdown tape with contents $\pi$ (i.e., $(M^*)_{\pi,\epsilon}$) by finite computations.

\begin{theorem}\label{thm:1}
        Let $\Sr$ be a complete starsemiring and $M\in \Snngam$ be a pushdown transition matrix.
        Then, for all $p \in \Gamma$ and $\pi \in \Gamma^*$,
        \begin{equation*}
                (M^*)_{p\pi, \epsilon} = (M^*)_{p,\epsilon}(M^*)_{\pi,\epsilon}\, .
        \end{equation*}
\end{theorem}

\begin{proof}
Since the case $\pi = \epsilon$ is trivial, we assume $\pi \in \Gamma^+$.
We obtain
\begin{align*}
(M^*)_{p\pi,\epsilon} & = \sum_{m\geq 0} (M^{m+1})_{p\pi,\epsilon}\\
& =  \sum_{m \geq 0} \sum_{\pi_1, \dots, \pi_m \in \Gamma^+}
M_{p\pi,\pi_1} M_{\pi_1,\pi_2} \dots M_{\pi_{m-1},\pi_m} M_{\pi_m, \epsilon} \\
& =  \Big( \sum_{m_1 \geq 0} \!\!\!\!\!
        \sum_{\ \ \ \rho_1, \dots, \rho_{m_1} \in \Gamma^+}
        \! \! \! \! \! M_{p\pi, \rho_1\pi} \dots M_{\rho_{m_1}\pi,\pi} \Big) \cdot
        \Big( \sum_{m_2 \geq 0} \! \! \! \! \! \sum_{\ \ \ \pi_1, \dots, \pi_{m_2} \in \Gamma^+} \! \! \! \!\! \!
        M_{\pi,\pi_1} \dots M_{\pi_{m_2},\epsilon} \Big)\\
& = \Big( \sum_{m_1 \geq 0} \sum_{\, \rho_1, \dots, \rho_{m_1} \in \Gamma^+} \!
        M_{p,\rho_1} \dots M_{\rho_{m_1}, \epsilon}\Big) (M^*)_{\pi, \epsilon} =
        (M^*)_{p,\epsilon} (M^*)_{\pi,\epsilon} \, .
\end{align*}

The summand for $m=0$ is $M_{p\pi, \epsilon}$;
the summand for $m_1 = 0$ is $M_{p\pi, \pi}$ or $M_{p,\epsilon}$;
the summand for $m_2=0$ is $M_{\pi,\epsilon}$.
In the third line in the first factor the pushdown contents are always of the form $\rho \pi,$ $\rho \in \Gamma^+$, except for the last move.
Hence, in the second factor the first move has to start with pushdown contents $\pi$ and it is the first time that the leftmost symbol of $\pi$ is read.
\end{proof}

Intuitively, the next lemma states that the infinite computations starting with $p_1 \dots p_k$ on the pushdown tape yield the same matrix $(M^\omega)_{p_1 \dots p_k}$ as summing up, for all $1\leq j\leq k$ the product of $(M^*)_{p_1 \dots p_{j-1},\epsilon}$ (i.e., emptying the pushdown tape with contents $p_1 \dots p_{j-1}$ by finite computations) with the matrix $(M^\omega)_{p_j}$ (i.e., the infinite computations starting with $p_j$ on the pushdown tape).

This means that in $p_1\dots p_k$ the pushdown symbols $p_1,\dots,p_{j-1}$ are emptied by finite computations and $p_j$ is chosen for starting the infinite computations. Clearly, $p_{j+1},\dots,p_k$ are not read.

\begin{lemma}\label{lem:3.3}
Let $(S,V)$ be a complete semiring-semimodule pair and let $M \in \Snngam$	be a pushdown transition matrix.
Then for all $p_1, \dots, p_k \in \Gamma$,
\begin{equation*}
(M^\omega)_{p_1 \dots p_k} = \sum_{1\leq j \leq k} (M^*)_{p_1, \dots, p_{j-1}, \epsilon} (M^\omega)_{p_j} \, .
\end{equation*}
\end{lemma}

\begin{proof}
        \begin{equation*}
        (M^\omega)_{p_1, \dots, p_k} = \sum_{{\rho_1, \rho_2, \dots} \in \Gamma^+} M_{p_1\dots p_k, \rho_1} M_{\rho_1,\rho_2}M_{\rho_2, \rho_3} \dots \, .
        \end{equation*}
        We partition the ``runs'' $(p_1 \dots p_k, \rho_1, \rho_2, \rho_3,\dots)$ into classes:
        \begin{itemize}
        \item class (1):
                there exist $\rho_i' \in \Gamma^+$, $i\geq 1$, such that $\rho_i = \rho_i'p_2\dots p_k$.
        \item class (j).(t), $k\geq 3$, $2 \leq j \leq k-1$, $t \geq 1$: $\rho_t=p_j \dots p_k$ and
                there exist $\rho_i' \in \Gamma^+$, for $1 \leq i \leq t-1$ and $i\geq t+1$, such that
                $\rho_i = \rho_i'p_j\dots p_k$ for
                $1 \leq i \leq t-1$, and $\rho_i = \rho_i' p_{j+1} \dots p_k$ for $i \geq t+1$.
        \item class (k).(t), $k \geq 2$, $t \geq 1$:
                $\rho_t = p_k$ and
                there exist $\rho_i' \in \Gamma^+$ for $1\leq i \leq t-1$, such that
                $\rho_i = \rho_i' p_k$.
        \end{itemize}
        Clearly, class (1) and class (j).(t), $2 \leq j \leq k$, $t\geq 1$ are pairwise disjoint.

        Intuitively, in the runs of \\
    class (1): $p_2$ is never read;\\
        class (j).(t), $2 \leq j \leq k-1$, $t\geq 1$: $p_{j+1}$
                is never read and $p_j$ is read in the $t$-th step;\\
        class (k).(t), $t \geq 1$: $p_k$ is read
                in the $t$-th step.\\

        We now compute for each class the value of
        \begin{equation*}
        S(1) = \sum_{(1)} M_{p_1 \dots p_k, \rho_1} M_{\rho_1, \rho_2}
                M_{\rho_2, \rho_3} \dots \, .
        \end{equation*}
        and
        \begin{equation*}
        S(j).(t) = \sum_{(j).(t)} M_{p_1 \dots p_k, \rho_1} M_{\rho_1\rho_2}
                M_{\rho_2, \rho_3} \dots, 2 \leq j \leq k, t \geq 1,
        \end{equation*}
        where $\sum_{(1)}$ and $\sum_{(j).(t)}$ means summation over all runs
        in the classes (1) and (j).(t), respectively.
        We obtain
        \begin{equation*}
        S(1) = \sum_{\rho_1', \rho_2', \dots \in \Gamma^+} M_{p_1, \rho_1'}
                M_{\rho_1', \rho_2'} M_{\rho_2', \rho_3'} \dots = (M^\omega)_{p_1} \, .
        \end{equation*}
        For $2 \leq j \leq k-1$, $t \geq 1$, we obtain
        \begin{align*}
        S(j).(t) = \ &\Big( \sum_{  \rho_1',\rho_2', \dots, \rho'_{t-1} \in \Gamma^+} \! \!
                M_{p_1 \dots p_{j-1}, \rho_1'} \dots M_{\rho'_{t-2}, \rho'_{t-1}}
                M_{\rho'_{t-1},\epsilon} \Big) \cdot \Big( \sum_{{\rho'_{t+1},\rho'_{t+2}, \dots} \in \Gamma^+} \!\!\!
                M_{p_j,\rho_{t+1}'} M_{\rho'_{t+1},\rho'_{t+2}}\dots \Big)\\
                = \ & (M^t)_{p_1\dots p_{j-1},\epsilon} (M^\omega)_{p_j} \, .
        \end{align*}
        For $t \geq 1$,
        \begin{align*}
        S(k).(t) =  &\ \Big( \sum_{  \rho_1',\rho_2', \dots, \rho'_{t-1} \in \Gamma^+} \! \!
        M_{p_1 \dots p_{k-1}, \rho_1'} \dots M_{\rho'_{t-2}, \rho'_{t-1}}
        M_{\rho'_{t-1},\epsilon} \Big) \cdot\Big(\sum_{{\rho_{t+1},\rho_{t+2}, \dots} \in \Gamma^+} \!\!\!
                M_{p_k,\rho_{t+1}} M_{\rho_{t+1},\rho_{t+2}}\dots \Big) \\
                = & \ (M^t)_{p_1\dots p_{k-1},\epsilon} (M^\omega)_{p_k} \, .
        \end{align*}
        Hence, we obtain
        \begin{align*}
        (M^\omega)_{p_1 \dots p_k} & = S(1) + \sum_{2 \leq j \leq k} \sum_{t \geq 1} S(j).(t) = (M^\omega)_{p_1} + \sum_{2 \leq j \leq k} (M^*)_{p_1 \dots p_{j-1},\epsilon} (M^\omega)_{p_j} \\
        & = \sum_{1 \leq j \leq k} (M^\star)_{p_1 \dots p_{j-1}, \epsilon}(M^\omega)_{p_j} \, .\qedhere
        \end{align*}
\end{proof}

Intuitively, our next theorem states that the infinite computations starting with $p$ on the pushdown tape yield the same matrix $(M^\omega)_p$ as summing up, for all $\pi=p_1\dots p_k$ and all $1\leq j\leq k$ the product of $M_{p,\pi}$ (i.e., changing the contents of the pushdown tape from $p$ to $\pi$) with the matrix $(M^*)_{p_1\dots p_{j-1},\epsilon}$ (i.e., emptying the pushdown tape with contents $p_1\dots p_{j-1}$ by finite computations) and eventually with the matrix $(M^\omega)_{p_j}$ (i.e., the infinite computations starting with $p_j$ on the pushdown tape).

This means that in $\pi$ the pushdown symbols $p_1,\dots,p_{j-1}$ are emptied by finite computations and $p_j$ is chosen for starting the infinite computations. Clearly, $p_{j+1}, \dots, p_k$ are not read.

\begin{theorem} \label{thm:3.4}
        Let $(S,V)$ be a complete semiring-semimodule pair and let
        $M \in \Snngam$ be a pushdown transition matrix.
        Then, for all $p \in \Gamma$,
        \begin{equation*}
        (M^\omega)_p = \sum_{p_1 \dots p_k \in \Gamma^+}
        M_{p, p_1\dots p_k} \sum_{1 \leq j \leq k}
        (M^*)_{p_1 \dots p_{j-1}, \epsilon} (M^\omega)_{p_j} \, .
        \end{equation*}
\end{theorem}

\begin{proof}
        We obtain, by Lemma \ref{lem:3.3}
        \begin{align*}
        \sum_{p_1 \dots p_k \in \Gamma^+}
        M_{p, p_1\dots p_k} \sum_{1 \leq j \leq k}
        (M^*)_{p_1 \dots p_{j-1}, \epsilon} (M^\omega)_{p_j}
        & =  \sum_{p_1 \dots p_k \in \Gamma^+}
        M_{p, p_1\dots p_k} (M^\omega)_{p_1\dots p_k} \\
        & = \sum_{\pi \in \Gamma^*} M_{p,\pi}(M^\omega)_\pi
        = (MM^\omega)= M^\omega \, .
        \end{align*}
\end{proof}

We define the matrices $(A_M)_{p,p'} \in \Snn$,
$M \in \Snngam$ a pushdown transition matrix,
$p,p' \in \Gamma$, by
\begin{equation*}
(A_M)_{p,p'} = \!\!\!\! \sum_{\substack{\pi=p_1 \dots p_k \in \Gamma^+\\
                p_j=p'}} M_{p,\pi} (M^*)_{p_1, \epsilon}\dots(M^*)_{p_{j-1},\epsilon} \, ,
\end{equation*}
and $A_M \in (\Snn)^{\Gamma \times \Gamma} $ by
$A_M = ((A_M)_{p,p'})_{p,p'\in \Gamma}$.
Whenever we use the notation $A_M$ we mean the matrix just defined.

\begin{theorem}\label{thm:3.5}
                Let $(S,V)$ be a complete semiring-semimodule pair and let
                $M \in \Snngam$ be a pushdown transition matrix.
                Then, for all $p \in \Gamma$,
                \begin{equation*}
                (M^\omega)_p = \sum_{p' \in \Gamma} (A_M)_{p,p'}
                        (M^\omega)_{p'}\, .
                \end{equation*}
\end{theorem}

\begin{proof}
        We obtain by Theorem \ref{thm:3.4}
        \begin{align*}
        \sum_{p' \in \Gamma} (A_M)_{p,p'}
        (M^\omega)_{p'} & = \sum_{p' \in \Gamma}
                \sum_{\pi= p_1 \dots p_k \in \Gamma^+}
                \sum_{1 \leq j \leq k}
                \delta_{p_j,p'}M_{p,\pi}(M^*)_{p_1\dots p_{j-1},\epsilon} (M^\omega)_{p'}\\
                & = \sum_{ p_1 \dots p_k \in \Gamma^+}
                M_{p,p_1\dots p_k} \sum_{1 \leq j \leq k}
                \sum_{p' \in \Gamma} \delta_{p_j,p'}
                (M^*)_{p_1\dots p_{j-1},\epsilon}(M^\omega)_{p'}\\
                & = \sum_{p_1 \dots p_k \in \Gamma^+}
                M_{p,p_1\dots p_k} \sum_{1 \leq j \leq k}
                (M^*)_{p_1\dots p_{j-1},\epsilon}(M^\omega)_{p_j}
                = (M^\omega)_p \, .
        \end{align*}
\end{proof}

When we say ``$G$ is the graph with adjacency matrix
$M \in (\Snn)^{\Gamma^* \times \Gamma^*}$'' then it means that $G$ is the graph with
adjacency matrix
$M' \in \Sr^{(\Gamma^*\times n)\times(\Gamma^*\times n)}$, where $M$
corresponds to $M'$ with respect to the canonical isomorphism
between $(\Snn)^{\Gamma^* \times \Gamma^*}$ and $\Sr^{(\Gamma^*\times n)\times(\Gamma^*\times n)}$.

Let now $M$ be a pushdown transition matrix and $0\leq l \leq n$.
Then $M^{\omega,l}$ is the column vector in $(V^n)^{\Gamma^*}$
defined as follows:
For $\pi \in \Gamma^*$ and $1 \leq i \leq n$,
let $((M^{\omega,l})_{\pi})_i$ be the sum of all weights of paths in the graph with
adjacency matrix $M$ that have initial vertex $(\pi,i)$ and visit vertices $(\pi',i')$, $\pi' \in \Gamma^*$, $1 \leq i' \leq l$, infinitely often.
Observe that $M^{\omega, 0}= 0$ and $M^{\omega, n}=M^\omega$.

Let $P_l = \{ (j_1, j_2, \dots) \in \{1, \dots, n\}^\omega \mid j_t \leq l \text{ for infinitely many } t \geq 1 \}$.

Then for $\pi \in \Gamma^+$, $1\leq j \leq n$, we obtain
\begin{equation*}
((M^{\omega,l})_\pi)_j =
\sum_{\pi_1,\pi_2,\dots \in \Gamma^+} \sum_{(j_1, j_2, \dots) \in P_l}
(M_{\pi,\pi_{1}})_{j,j_1}(M_{\pi_1,\pi_2})_{j_1,j_2}(M_{\pi_2,\pi_3})_{j_2,j_3} \dots \, .
\end{equation*}

%

\begin{lemma} \label{lem:3.7}
        Let $(S,V)$ be a complete semiring-semimodule pair and let
        $M \in \Snngam$ be a pushdown transition matrix.
        Then, for all $p_1, \dots, p_k \in \Gamma$,
        $0 \leq l \leq n$,
        \begin{equation*}
        (M^{\omega, l})_{p_1 \dots p_k} =
        \sum_{1 \leq j \leq k} (M^*)_{p_1 \dots p_{j-1},\epsilon} (M^{\omega,l})_{p_j}.
        \end{equation*}
\end{lemma}

\begin{proof}
        By the proof of Lemma
        \ref{lem:3.3} and the following summation identity:
Assume that $A_1, A_2, \dots$ are matrices in $\Snn$. Then, for $0 \leq l \leq n$, $1 \leq j \leq n$, and $m \geq 1$,
\begin{align*}
\sum_{(j_1,j_2, \dots) \in P_l} & (A_1)_{j,j_1}(A_2)_{j_1,j_2\dots} =
\sum_{1 \leq j_1, \dots, j_m \leq n} (A_1)_{j, j_1}\dots
(A_m)_{j_{m-1},j_m} \sum_{(j_{m+1},j_{m+2}, \dots)\in P_l}
(A_{m+1})_{j_m,j_{m+1}} \dots \, .
\end{align*}
\end{proof}

Theorem \ref{thm:3.8} generalizes Theorem \ref{thm:3.5} from $M^{\omega,n}$ to
$M^{\omega, l}$, $0 \leq l \leq n$.

\begin{theorem} \label{thm:3.8}
                Let $(S,V)$ be a complete semiring-semimodule pair and let
                $M \in \Snngam$ be a pushdown transition matrix.
                Then, for all $p \in \Gamma$, $0 \leq l \leq n$,
                \begin{equation*}
                (M^{\omega,l})_p = \sum_{p' \in \Gamma} (A_M)_{p,p'}
                (M^{\omega,l})_{p'}\, .
                \end{equation*}
\end{theorem}

\section{Algebraic systems and $\omega$-pushdown automata} \label{sec:3}

In this section, we define $\omega$-pushdown automata
and show that for an $\omega$-pushdown automaton $\Pcc$
there exists an algebraic system over a quemiring such that the behavior $\bhvr{\Pcc}$
of $\Pcc$ is a component of a solution of this system.

For the definition of an $\Sr'$-algebraic system over a quemiring  $\Sr \times V$
we refer the reader to \cite{MAT}, page 136, and for the definition of
quemirings to \cite{MAT}, page 110.
Here we note that a quemiring   $T$
is isomorphic to a quemiring $\Sr \times V$ determined by the semiring-semimodule pair  $(\Sr,V)$,
cf. \cite{MAT}, page 110.

In the sequel, $(S,V)$ is a complete semiring-semimodule pair and $\Sr'$ is a subset of $S$ containing $0$ and $1$.
Let $M \in \Spnngam$ be a pushdown matrix.
Consider the $\Spnn$-algebraic system over the complete semiring-semimodule pair $(\Snn,\Sn)$, i.e., over the quemiring $\Snn \times \Sn$,
\begin{equation}\label{sys:1}
y_p = \sum_{\pi \in \Gamma^*} M_{p, \pi} y_{\pi} \, , \, p \in \Gamma \, .
\end{equation}
(See Section 5.6 of \'Esik, Kuich \cite{MAT}.)
The variables of this system \eqref{sys:1} are
$y_p, p \in \Gamma$, and
$y_\pi, \pi \in \Gamma^*$, is defined by
$y_{p\pi}= y_p y_\pi $ for $p \in \Gamma$,
$\pi \in \Gamma^*$ and $y_\epsilon = \epsilon$.
Hence, for $\pi = p_1 \dots p_k$,
$y_\pi = y_{p_1} \dots y_{p_k}$.
The variables $y_p$ are variables for $(\Snn, \Sn)$.

Let $x = (x_p)_{p \in \Gamma}$, where
$x_p$, $p \in \Gamma$, are variables for $\Snn$.
Then, for $p \in \Gamma$, $\pi = p_1 p_2 \dots p_k$,
$(M_{p,\pi} y_\pi)_x$ is defined to be
\begin{equation*}
(M_{p,\pi}y_\pi)_x = (M_{p,\pi} y_{p_1}\dots y_{p_k})_x =
M_{p,\pi}z_{p_1} + M_{p,\pi}x_{p_1}z_{p_2} + \dots + M_{p,\pi} x_{p_1} \dots x_{p_{k-1}}z_{p_k}.
\end{equation*}
Here $z_p$, $p \in \Gamma$, are variables for $\Sn$.

We obtain, for $p \in \Gamma$, $\pi = p_1 \dots p_k$,
\begin{align*}
(M_{p, \pi} y_\pi)_x & =
\sum_{p'\in\Gamma}\sum_{\substack{\pi=p_1 \dots p_k \in \Gamma^+\\
                p_j=p'}} M_{p, \pi}
x_{p_1}\dots x_{p_{j-1}} z_{p'} \\
& = \sum_{\pi= p_1 \dots p_k \in \Gamma^+}
M_{p,\pi} \sum_{1 \leq j \leq k} x_{p_1}\dots
x_{p_{j-1}} z_{p_j} \, .
\end{align*}

The system \eqref{sys:1} induces the following
mixed $\omega$-algebraic system:

\begin{align}
x_p & = \sum_{\pi \in \Gamma^*} M_{p \pi} x_\pi \, , \, p \in \Gamma, \label{sys:2} \\
z_p & = \sum_{\pi\in\Gamma^*} (M_{p,\pi}y_\pi)_{(x_p)_{p \in \Gamma}} =
\sum_{p'\in\Gamma} \sum_{\substack{\pi=p_1 \dots p_k \in \Gamma^+\\
                p_j=p'}} M_{p, \pi}
x_{p_1}\dots x_{p_{j-1}} z_{p'} \, . \label{sys:3}
\end{align}

Here \eqref{sys:2} is an $\Spnn$-algebraic system
over the semiring $\Snn$ (see Section 2.3 of \'Esik, Kuich \cite{MAT}) and \eqref{sys:3} is an
$\Snn$-linear system over the semimodule $\Sn$
(see Section 5.5 of \'Esik, Kuich \cite{MAT}).

In the classical theory of automata and formal languages, equation \eqref{sys:2} plays a crucial role in the transition from pushdown automata to
context-free grammars. It is, in the form of matrix notation, the well-known triple construction.
(See Harrison \cite{Har}, Theorem 5.4.3;
Bucher, Maurer \cite{17}, S\"atze 2.3.10, 2.3.30; Kuich, Salomaa \cite{88}, pages 178, 306; Kuich \cite{78}, page 642;
\'Esik, Kuich \cite{MAT}, pages 77, 78.)

By Theorem 5.6.1 of \'Esik, Kuich \cite{MAT},
$(A,U) \in ((\Snn)^\Gamma,(\Sn)^\Gamma)$ is a solution of
\eqref{sys:1} iff $A$ is a solution of \eqref{sys:2} and $(A,U)$ is a solution of \eqref{sys:3}. We now compute such solutions
$(A,U)$.

\begin{theorem}\label{thm:9}
        Let $S$ be a complete starsemiring and
        $M \in \Spnngam$ be a pushdown transition matrix.
        Then $((M^*)_{p,\epsilon})_{p \in \Gamma}$
        is a solution of \eqref{sys:2}.
\end{theorem}
\begin{proof}
        By Theorem \ref{thm:1}.
\end{proof}

We now substitute in \eqref{sys:3} for $(x_p)_{p\in\Gamma}$ the solution
$((M^*)_{p,\epsilon})_{p \in \Gamma}$ of \eqref{sys:1}
and obtain the $\Spnn$-linear system
\eqref{sys:4} over the semimodule $\Sn$
\begin{equation}\label{sys:4}
z_p = \sum_{p' \in \Gamma} (A_M)_{p,p'} z_{p'}  ,\, p \in \Gamma \, .
\end{equation}

\begin{theorem}
        Let $(S,V)$ be a complete semiring-semimodule pair and
        $M \in \Spnngam$ be a pushdown transition matrix.
        Then, for all $0 \leq l \leq n$,
        $((M^{\omega,l})_p)_{p \in \Gamma}$ is a solution
        of \eqref{sys:4}.
\end{theorem}

\begin{proof}
        By Theorem \ref{thm:3.8}.
\end{proof}

\begin{corollary}\label{cor:11:4.3}
        Let $(S,V)$ be a complete semiring-semimodule pair and
        $M \in \Spnngam$ be a pushdown transition matrix.
        Then, for all $0 \leq l \leq n$,
        \begin{equation*}
        (((M^*)_{p, \epsilon})_{p \in \Gamma}, ((M^{\omega,l})_p)_{p \in \Gamma})
        \end{equation*}
        is a solution of \eqref{sys:1}.
\end{corollary}

We can write the system \eqref{sys:4} in matrix notation in the form
\begin{equation}\label{sys:5}
z = A_M z
\end{equation}
with column vector $z = (z_p)_{p \in \Gamma}$.

\begin{corollary}
        Let $(S,V)$ be a complete semiring-semimodule pair and
        $M \in \Spnngam$ be a pushdown transition matrix.
        Then for all $0 \leq l \leq n$,
        $((M^{\omega,l})_p)_{p \in \Gamma}$
        is a solution of \eqref{sys:5}.
\end{corollary}

We now introduce pushdown automata and $\omega$-pushdown automata
(see Kuich, Salomaa \cite{88}, Kuich \cite{78}, Cohen, Gold \cite{23}).

 Let $\Sr$ be a complete semiring and
 $\Sr' \subseteq \Sr$ with $0,1 \in \Sr'$.
 An \emph{$\Sr'$-pushdown automaton over $\Sr$}
 $$\Pmc = (n,\Gamma,I,M,P,p_0)$$
 is given by
 \begin{itemize}
        \item[(\textit{i})] a finite set of \emph{states} $\{1,\dots,n\}$,
        $n\geq 1$,
        \item[(\textit{ii})] an alphabet $\Gamma$ of \emph{pushdown symbols},
        \item[(\textit{iii})] a \emph{pushdown transition matrix}
        $M \in  (\Spnn)^{\Gamma^* \times \Gamma^*}$,
        \item[(\textit{iv})] an \emph{initial state vector} $I \in {S'}^{1 \times n}$,
        \item[(\textit{v})] a \emph{final state vector} $P \in {S'}^{n \times 1}$,
        \item[(\textit{vi})] an \emph{initial pushdown symbol} $p_0 \in \Gamma$,
 \end{itemize}

 The \emph{behavior $\bhvr{\Pmc}$ of $\Pmc$} is an element of
 $\Sr$ and is defined by
 $\bhvr{\Pmc} = I(M^*)_{p_0,\epsilon}P $.

 For a complete semiring-semimodule pair
 $(S,V)$,
 an $\Sr'$-$\omega$-pushdown automaton
 (over $(S,V)$)
 \begin{equation*}
 \Pmc = (n, \Gamma, I, M, P, p_0, l)
 \end{equation*}
 is given by an $\Sr'$-pushdown automaton
 $(n, \Gamma, I,M, P, p_0)$ and an $l \in \{0,\dots,n\}$ indicating that the states $1,\dots, l$
 are \emph{repeated states}.

 The \emph{behavior} $\bhvr{\Pmc}$ of the $\Sr'$-$\omega$-pushdown automaton $\Pmc$ is defined by
 \begin{equation*}
 \bhvr{\Pmc} = I(M^*)_{p_0,\epsilon} P + I(M^{\omega,l})_{p_0}\, .
 \end{equation*}

 Here $I(M^*)_{p_0,\epsilon}P$ is the behavior of the
 $\Sr'$-$\omega$-pushdown automaton
 $\Pmc_1=(n,\Gamma,I,M,P,p_0,0)$ and
 $I(M^{\omega,l})_{p_0}$ is the behavior of the $\Sr'$-$\omega$-pushdown automaton
 $\Pmc_2= (n,\Gamma,I,M,0,p_0,l)$.
 Observe that $\Pmc_2$ is an automaton with the
 B\"{u}chi acceptance condition:
 if $G$ is the graph with adjacency matrix $M$, then only paths that visit
 the repeated states $1,\dots,l$ infinitely often contribute to
 $\bhvr{\Pmc_2}$.
 Furthermore, $\Pmc_1$ contains no repeated states and behaves like an ordinary $\Sr'$-pushdown automaton.

 \begin{theorem}
        Let $(S,V)$ be a complete semiring-semi\-module pair and let
        $\Pmc = (n,\Gamma, I, M, P,\allowbreak p_0,l)$ be an
        $S'$-$\omega$-pushdown automaton over $(S,V)$.
        Then $(\bhvr{\Pmc}, (((M^*)_{p,\epsilon})_{p\in \Gamma},((M^{\omega,l})_{p})_{p\in \Gamma}))$
         is a solution of the
        $\Spnn$-algebraic system
        \begin{equation*}
        y_0 = Iy_{p_0}P, y_p = \sum_{\pi \in \Gamma^*}
        M_{p,\pi}y_\pi, \, p \in \Gamma
        \end{equation*}
        over the complete semiring-semimodule pair
        $(\Snn,\Sn)$.
 \end{theorem}

 \begin{proof}
        By Corollary \ref{cor:11:4.3},
        $(((M^*)_{p,\epsilon})_{p\in \Gamma},((M^{\omega,l})_{p})_{p\in \Gamma})$
        is a solution of the second equation.
        Since
        \begin{equation*}
        I(((M^*)_{p_0,\epsilon}),((M^{\omega,l})_{p_0}))P =
        (I(M^*)_{p_0,\epsilon}P, I(M^{\omega, l})_{p_0}) =
        \bhvr{\Pmc} \, ,
        \end{equation*}
        $(\bhvr{\Pmc}, (((M^*)_{p,\epsilon})_{p \in \Gamma}, ((M^{\omega,l})_p)_{p\in \Gamma}))$
        is a solution of the given $\Spnn$-algebraic system.
 \end{proof}

 Let $\Sr$ be a complete star-omega semiring and $\Sigma$ be an alphabet.
 Then by Theorem 5.5.5 of \'Esik, Kuich \cite{MAT},
 $(S \llss, S \llso)$ is a complete semiring-semimodule pair.
 Let $\Pmc = (n, \Gamma, M, I, P, p_0,l)$ be an
 $\ssigep$-$\omega$-pushdown automaton over
 ${(S \llss,}$ ${S \llso)}$.
 Consider the algebraic system over the complete
 semiring-semimodule pair
 $({(S \llss)^{n\times n}}, {(S \llso)^{n}})$
 \begin{equation} \label{sys:6}
 y_0 = I y_{p_0} P, y_p = \sum_{\pi \in \Gamma^*} M_{p,\pi}y_\pi, p \in \Gamma
 \end{equation}
 and the mixed algebraic system \eqref{sys:7}
 over $((S \llss)^{n \times n}, (S \llso)^n)$
 induced by \eqref{sys:6}
 \begin{equation}\label{sys:7}
 \begin{aligned}
 x_0 & = I x_{p_0}P, x_p = \sum_{\pi = p_1 \dots p_k \in \Gamma^*} M_{p,\pi} x_{p_1} \dots x_{p_k}, p \in \Gamma \, ,\\
 z_0 & = I z_{p_0}, z_p = \sum_{\pi = p_1 \dots p_k \in \Gamma^+} M_{p,\pi} \sum_{1 \leq j \leq k} x_{p_1} \dots x_{p_{j-1}} z_{p_j}, p \in \Gamma \,.
 \end{aligned}
 \end{equation}

 \begin{corollary}
         Let $(S,V)$ be a complete semiring-semimodule pair,
        $\Sigma$ be an alphabet and
        $\Pmc = (n,\Gamma, M,I,P,\allowbreak p_0,l)$ be an
        $\ssigep$-$\omega$-pushdown automaton over
        $(S \llss, S \llso)$.

        Then $(I(M^*)_{p_0,\epsilon}P,
        ((M^*)_{p,\epsilon})_{p \in \Gamma},\allowbreak
        I(M^{\omega, l})_{p_0},
        ((M^{\omega, l})_p)_{p \in \Gamma})$
        is a solution of \eqref{sys:7}.
        It is called solution of order $l$.
 \end{corollary}

 Let now in \eqref{sys:7}
 \begin{equation*}
 x = ( [i,p,j])_{1 \leq i,j \leq n} , p \in \Gamma,
 \end{equation*}
 be $n \times n$-matrices of variables and
 \begin{equation*}
 z = ([i,p])_{1 \leq i \leq n}, p\in \Gamma
 \end{equation*}
 be $n$-dimensional column vectors of variables.
 If we write the mixed algebraic system \eqref{sys:7}
 component-wise, we obtain a mixed algebraic system over
 $({(\Sr\ll \Sigma^* \gg),} {(\Sr \ll \Sigma^\omega \gg)} )$
 with variables $\ipj$ over $\Sr\llss$, where $p \in \Gamma$, $1 \leq i,j \leq n$, and
 variables $[i,p]$ over $\Sr \ll \Sigma^\omega \gg$,
 where $p \in \Gamma$, $1 \leq i \leq n$.

 Writing the mixed algebraic system \eqref{sys:7}
 component-wise, we obtain the system \eqref{sys:8}:
 \begin{equation}\label{sys:8}
  \begin{aligned}
  &x_0  = \sum_{1 \leq m_1 , m_2 \leq n} I_{m_1}
 [m_1, p_0,m_2] P_{m_2} ,\\
  &\ipj=  \sum_{k \geq 0} \sum_{p_1,\dots,p_k \in \Gamma}
 \sum_{1 \leq m_1,\dots,m_k \leq n} (M p,p_1\dots p_k)_{i,m_1} [m_1, p_1,m_2][m_2,p_2,m_3] \dots
 [m_k,p_k,j] , \\
 &{\hspace{11cm} p \in \Gamma, 1 \leq i,j \leq n ,} \\
  &z_0 = \sum_{1 \leq m \leq n} I_m [m,p_0]\\
  &[i,p] = \! \sum_{k \geq 1} \sum_{p_1, \dots, p_k \in \Gamma} \sum_{1 \leq j \leq k} \sum_{1 \leq m_1, \dots, m_j \leq n}\hspace{-0.5cm} (M p,p_1\dots p_k)_{i,m_1} [m_1, p_1,m_2]
 \dots[m_{j\text{-}1},p_{j\text{-}1},m_j][m_{j},p_j],\\
 & \hspace{11.2cm} \, p \in \Gamma, 1 \leq i \leq n \, .
 \end{aligned}
 \end{equation}

 \begin{theorem}\label{thm:15}
         Let $(S,V)$ be a complete semiring-semimodule pair and
        $\Pmc= (n, \Gamma, M, I, p_0, \allowbreak P,l)$ be a
        $\Sr'$-$\omega$-pushdown automaton.
        Then
        \begin{equation*}
        (I(M^*)_{p_0,\epsilon} P,(((M^*)_{p,\epsilon})_{i,j})_{p \in \Gamma, 1 \leq i,j \leq n},
        I(M^{\omega,l})_{p_0}, ((M^{\omega, l})_p)_i)_{p\in \Gamma, 1 \leq i \leq n}
        \end{equation*}
        is a solution of the system \eqref{sys:8} called solution of order $l$ with
        $\bhvr{\Pmc} = (I(M^*)_{p_0, \epsilon} P , I(M^{\omega, l})_{p_0})$.
 \end{theorem}

\section{Mixed algebraic systems and mixed context-free grammars} \label{sec:4}

In this section we associate a mixed context-free
grammar with finite and infinite derivations to the
algebraic system \eqref{sys:8}.
The language generated by this mixed context-free
grammar is then the behavior $\bhvr{\Pmc}$ of the
$\omega$-pushdown automaton $\Pmc$.
The construction of the mixed context-free grammar from
the $\omega$-pushdown automaton $\Pmc$ is a generalization of
the well known triple construction and is called now
\emph{triple-pair construction for $\omega$-pushdown
        automata}.
We will consider the commutative complete star-omega semirings
$\B = ( \{0,1\}, \vee, \land, *,0,1)$ with
$0^* =1^*=1$ and $\N^\infty = (\N \cup \{\infty\}, +, \cdot, ^*, 0,1)$
with $0^* =1$ and $a^* = \infty$ for $a \neq 0$.

If $\Sr = \mathbb{B}$ or $\Sr = \N^\infty$ and
$0 \leq l \leq n$, then we associate to the mixed algebraic
system \eqref{sys:8} over
$((\Sr\ll \Sigma^* \gg), (\Sr \ll \Sigma^\omega \gg) )$, and hence to the $\omega$-pushdown automaton $\Pmc = (n, \Gamma, I, M, P, p_0, l)$,
the \emph{mixed context-free grammar}
\begin{equation*}
G_l \ = \ (X,Z,\Sigma, P_X, P_Z, x_0, z_0, l) \, .
\end{equation*}
( See also \'Esik, Kuich \cite[page 139]{MAT}.)
Here

\begin{itemize}
        \item[(i)] $X=\{x_0\} \cup \{[i,p,j]\mid 1\leq i,j\leq n,\, p \in \Gamma\}$ is a set of \emph{variables for finite derivations};
        \item[(ii)] $Z = \{z_0\} \cup \{[i,p] \mid 1 \leq i \leq n, \, p \in \Gamma \}$ is a set of \emph{variables for infinite derivations};
        \item[(iii)] $\Sigma$ is an alphabet of \emph{terminal symbols};
        \item[(iv)] $P_X$ is a finite set of \emph{productions for finite derivations} given below;
        \item[(v)] $P_Z$ is a finite set of \emph{productions for infinite derivations} given below;
        \item[(vi)] $x_0$ is the \emph{start variable for finite derivations};
        \item[(vii)] $z_0$ is the \emph{start variable for infinite derivations};
        \item[(viii)] $\{[i,p] \mid 1 \leq i \leq l, \, p \in \Gamma \}$ is the set of
        \emph{repeated variables for infinite derivations}.
\end{itemize}
In the definition of $G_l$
the sets $P_X$ and $P_Z$ are as follows:
\begin{align*}
P_X = \  & \{ x_0 \to a_1 \mpom a_2 \mid \\
& \ \ 1 \leq m_1,m_2 \leq n, (I_{m_1},a_1) \neq 0 , (P_{m_2},a_2) \neq 0,
a_1,a_2 \in \Sigma \cup \{\epsilon\}\}\ \cup\\
& \{\ipj \to a \mpim [m_2,p_2,m_3] \dots [m_k,p_k,j]\mid p \in \Gamma,  1 \leq i,j \leq n, k\geq 0,\\
& \ \ p_1,\dots,p_k \in \Gamma, 1 \leq m_1,\dots,m_k \leq n,
((M_{p,p_1\dots p_k})_{i,m_1},a)\neq 0, a \in \Sigma \cup \{\epsilon\}  \} \ , \\
P_Z = \ & \{z_0 \to a [m,p_0] \mid 1 \leq m \leq n, (I_m,a) \neq 0, a \in \Sigma \cup \{\epsilon\} \} \ \cup\\
& \{ \ip \to a \mpim \dots [m_{j-1},p_{j-1},m_j][m_j,p_j]\mid p, p_1, \dots, p_k \in \Gamma, 1 \leq i \leq n, \\
& \ \  k\geq 1, 1 \leq j \leq k, 1 \leq m_1,\dots,m_j\leq n,
((M_{p,p_1\dots p_k})_{i,m_1},a) \neq 0, a \in \Sigma \cup \{\epsilon\}\} \ .
\end{align*}
For the remainder of this section, $\Pmc$ always denotes the $\omega$-pushdown automaton  $\Pmc = (n, \Gamma, I, M, P, p_0, l)$. Especially this means that $l$ is a fixed parameter.
Observe that $((M_{p,p_1\dots p_k})_{i,m_1},a)\neq 0$ iff
$(m_1,p_k\dots p_1) \in \delta(i,a,p)$ in the usual $\delta$-notation for the transition function of a classical pushdown automaton. (See Harrison \cite{Har} and Kuich \cite{78} pages 638/639.)
Here we have to reverse $p_1 \dots p_k$ since the pushdown
tape of classical pushdown automata has its rightmost element
as top element.

A \emph{finite leftmost derivation}
$\alpha_1 \derls \alpha_2$, where
$\alpha_1, \alpha_2 \in (X \cup \Sigma)^*$, by productions
in $P_X$ is defined as usual.
An \emph{infinite (leftmost) derivation}
$\pi : z_0 \derl^{\omega} w$, for $z_0 \in Z, w \in \Sigma^\omega$, is defined as follows:
\begin{align*}
\pi:\ & z_0 \derl \alpha_0 [{i_0},p_0] \derls w_0  [i_0,p_0] \derl w_0
\alpha_1 [i_1,p_1] \derls w_0 w_1  [{i_1},p_1]  \derl \dots \\
& \derls w_0w_1 \dots w_m  [{i_m},p_m]  \derl
w_0 w_1 \dots w_m \alpha_{m+1}
[{i_{m+1}},p_{m+1}]  \derls \dots \, ,
\end{align*}
where $z_0 \to \alpha_0 [{i_0},p_0] ,
[{i_0},p_0] \to \alpha_1  [{i_1},p_1] , \dots, [{i_m},p_{m}] \to \alpha_{m+1}[{i_{m+1}},p_{m+1}], \dots$
are productions in $P_Z$ and $w = w_0 w_1 \dots w_m \dots$.

We now define an infinite derivation
$\pi_l : z_0 \derl^{\omega,l} w$ for $0 \leq l \leq n$, $z_0 \in Z$, $w \in \Sigma^\omega$:
We take the above definition for
$\pi:z_0 \Rightarrow_L^\omega w$
and consider the sequence of the first elements $i$ of the
triple variables $[i,p,j]$ of $X$ that are rewritten in the finite leftmost
derivation $\alpha_m \derls w_m$, $m \geq 0$.
Assume this sequence is $i_m^1, i_m^2, \dots, i_m^{t_m}$ for some $t_m$, $m \geq 1$.
Then, to obtain $\pi_l$ from $\pi$, the condition
$i_0, i_1^1, i_1^2 \dots, i_1^{t_1}, i_1, i_2^1, \dots,\allowbreak i_2^{t_2}, i_2, \dots, i_m, i_{m+1}^1, \dots, i_{m+1}^{t_{m+1}}, i_{m+1}, \dots \in P_l$ has to be satisfied.

Then we define
\begin{equation*}
L(G_l) =  \{w \in \Sigma^* \mid x_0 \derls w \} \ \cup \ \{ w \in \Sigma^\omega \mid \pi : z_0 \derl^{\omega,l} w \} \, .
\end{equation*}
Observe that the construction of $G_l$ from $\Pmc$ is
nothing else than a generalization of the triple
construction for $\omega$-pushdown automata, since the construction
of the context-free grammar $G = (X, \Sigma, P_X,x_0)$
 is the triple construction.
(See  Harrison \cite{Har}, Theorem 5.4.3;
Bucher, Maurer \cite{17}, S\"atze 2.3.10, 2.3.30; Kuich, Salomaa \cite{88}, pages 178, 306; Kuich \cite{78}, page 642;
\'Esik, Kuich \cite{MAT}, pages 77, 78.)

We call the construction of the mixed context-free grammar
$G_l$ from $\Pmc$
the \emph{triple-pair construction for $\omega$-pushdown automata}.
This is justified by the definition of the sets of variables
$\{ \ipj \mid 1 \leq i,j, \leq n, p \in \Gamma\}$ and
$\{[i,p] \mid 1 \leq i \leq n, p \in \Gamma\}$ of $G_l$ and by the forthcoming Corollary \ref{cor:4.2}.

In the next theorem we use the isomorphism
between ${\B \llss } \times {\B \llso}$ and
$2^{\Sigma^*} \times 2^{\Sigma^\omega}$.

\begin{theorem}\label{thm:4.1}
        Assume that $(\sigma, \tau)$ is the solution of order
        $l$ of the mixed algebraic system \eqref{sys:8} over
        $(\B \llss , \B \llso)$ for $k \in \{0, \dots, n\}$.
        Then
        \begin{equation*}
        L(G_l) \ = \ \sigma_{x_0} \cup \tau_{z_0} \, .
        \end{equation*}
\end{theorem}

\begin{proof}
        By Theorem \rom{4}.1.2 of Salomaa, Soittola \cite{SalSoi} and by Theorem \ref{thm:15},
        we obtain $\sigma_{x_0} = \{w \in \Sigma^* \mid x_0 \derls w\}$.
        We now show that $\tau_{z_0}$ is generated by the infinite derivations $\derl^{\omega,l}$ from $z_0$.
        First observe that the rewriting by the
        typical $[i,p,j]$- and $[i,p]$- production
        corresponds to the situation that
        in the graph of the $\omega$-pushdown automaton $\Pmc$ the edge from $(p\rho,i)$
        to $(p_1\dots p_j\rho, j)$, $\rho \in \Gamma^*$, is passed
        after the state $i$ is visited.
        The first step of the infinite derivation
        $\pi_l$ is given by $z_0 \derl \alpha_0 [i_0,p]$ and indicates that the path in the graph of $\Pmc$ corresponding to $\pi_l$ starts in state $i_0$.
        Furthermore, the sequence of the first elements of variables that are rewritten in $\pi_l$, i.e., $i_0, i_1^1, \dots, i_1^{t_1}, i_1, i_2^1, \dots, i_2^{t_2}, i_2, \dots, i_m, i_{m+1}^1, \dots, i_{m+1}^{t_{m+1}}, i_{m+1}, \dots$ indicates that the path in the graph of $\Pmc$ corresponding to $\pi_l$ visits these states.
        Since this sequence is in $P_l$ the corresponding path
        contributes to $\bhvr{\Pmc}$. Hence, by Theorem \rom{4}.1.2 of Salomaa, Soittola \cite{SalSoi} and Theorem \ref{thm:15} for the finite leftmost derivations
        $\alpha_m \derls w_m$, $m \geq 1$, and by Theorem 5.5.9 of \'Esik, Kuich \cite{MAT} and Theorem \ref{thm:15} for the infinite derivation $[i_0,p_0]\Rightarrow \alpha_1 [i_1,p_1] \Rightarrow \alpha_1 \alpha_2 [i_2,p_2] \Rightarrow \dots \Rightarrow \alpha_1 \alpha_2 \dots \alpha_m [i_m,p_m]\Rightarrow \dots$
        we obtain
        \begin{equation*}
        \tau_{z_0} = \{ w \in \Sigma^\omega \mid
        \pi: z_0 \derl^{\omega,l} w\} \, .
        \end{equation*}
\end{proof}

\begin{corollary}\label{cor:4.2}
        Assume that
        the mixed context free grammar $G_l$ associated to the mixed algebraic system \eqref{sys:8} is constructed from the
        $\B \langle \Sigma \cup \{\epsilon\} \rangle$-%
        $\omega$-pushdown automaton $\Pmc$.
        Then
        \begin{equation*}
        L(G_l) = \bhvr{\Pmc} \, .
        \end{equation*}
\end{corollary}
\begin{proof}
        By Theorems \ref{thm:15} and \ref{thm:4.1}.
\end{proof}

For the remainder of this section our basic semiring is $\N^\infty$,
which allows us to draw some stronger conclusions.

\begin{theorem}\label{thm:16}
        Assume that $(\sigma, \tau)$ is the
        solution of order $l$ of the mixed algebraic system
        \eqref{sys:8} over
        $(\N^\infty \llss , \N^\infty \llso)$ where the entries of $I,M,P$ are in
        $\{0,1\} \langle \Sigma \cup \{\epsilon\}\rangle$.
        Denote by $d(w)$, for $w \in \Sigma^*$, the number
        (possibly $\infty$) of distinct finite leftmost
        derivations of $w$ from $x_0$ with respect to $G_l$;
        and by $c(w)$, for $w \in \Sigma^\omega$, the number
        (possibly $\infty$) of  distinct infinite leftmost
        derivations $\pi$  of $w$
        from $z_0$ with respect to $G_l$.
        Then
        \begin{equation*}
        \sigma_{x_0} = \sum_{w \in \Sigma^*} d(w)w \qquad \text{\ and \ } \qquad \tau_{z_0} = \sum_{w \in \Sigma^\omega} c(w)w \, .
        \end{equation*}
\end{theorem}

\begin{proof}
        The proof of Theorem \ref{thm:16} is identical to the proof of Theorem \ref{thm:4.1} with the exceptions that
        Theorem \rom{4}.1.2 of Salomaa, Soittola \cite{SalSoi} is replaced by Theorem \rom{4}.1.5 and Theorem 5.5.9 of
        \'Esik, Kuich \cite{MAT} is replaced by Theorem 5.5.10.
\end{proof}

In the forthcoming Corollary \ref{cor:4.4} we consider,
for a given $\{0,1\}\langle \Sigma \cup \{\epsilon \}\rangle$-$\omega$-pushdown automaton
$\Pmc = (n,\Gamma, I, M, P, p_0,l)$ the number of distinct computations from an initial instantaneous description $(i,w,p_0)$ for $w \in \Sigma^*$, $I_i \neq 0$, to an
accepting instantaneous description $(j, \epsilon, \epsilon)$,
with $P_j \neq 0$, $i,j \in \{0, \dots, n\}$.

Here $(i,w,p_0)$ means that $\Pmc$ starts in the initial state $i$ with $w$ on its input tape and $p_0$ on its
pushdown tape;
and $(j,\epsilon,\epsilon)$ means that $\Pmc$ has entered the final state $j$ with empty input tape and empty pushdown tape.

Furthermore, we consider the number of distinct infinite computations starting in an initial instantaneous description
$(i,w,p_0)$ for $w \in \Sigma^\infty$, $I_i \neq 0$.

\begin{corollary}\label{cor:4.4}
        Assume that
        the mixed context-free grammar $G_l$ associated to the mixed algebraic system \eqref{sys:8} is constructed from the $\{0,1\}\langle \Sigma \cup \{\epsilon \}\rangle$-$\omega$-pushdown automaton $\Pmc$.
        Then the number (possibly $\infty$) of distinct finite leftmost derivations of $w$, $w \in \Sigma^*$, from $x_0$
        equals the number of distinct finite computations from
        an initial instantaneous description for $w$ to an accepting instantaneous description;
        moreover, the number  (possibly $\infty$) of distinct infinite (leftmost) derivations of $w$, $w \in \Sigma^\omega$, from $z_0$ equals the number of distinct
        infinite computations starting in an initial instantaneous description for $w$.
\end{corollary}

\begin{proof}
        By Corollary 6.11 of Kuich \cite{78} and the definition of infinite derivations with respect to $G_l$.
\end{proof}

The context-free grammar $G_l$ associated to \eqref{sys:8}
is called \emph{unambiguous} if each $w \in L(G_l)$,
$w \in \Sigma^*$ has a unique finite leftmost derivation
and each $w \in L(G_l)$, $w \in \Sigma^\omega$,
has a unique infinite (leftmost) derivation.

An $\Nlse$-$\omega$-pushdown automaton $\Pmc$ is called
\emph{unambiguous} if $(\bhvr{\Pmc},w) \in \{0,1\}$ for each $w \in \Sigma^* \cup \Sigma^\omega$.

\begin{corollary}
        Assume that
        the mixed context-free grammar $G_l$ associated to the mixed algebraic system \eqref{sys:8} is constructed from the  $\{0,1\}\langle \Sigma \cup \{\epsilon \}\rangle$-$\omega$-pushdown automaton $\Pmc$.
        Then $G_l$ is unambiguous iff $\bhvr{\Pmc}$
        is unambiguous.
\end{corollary}

\bibliographystyle{eptcs}
\bibliography{drostekuich}
\end{document}